\documentclass[reqno]{amsart}
\usepackage{amsmath,amsthm,amsfonts,amssymb}
\usepackage{graphicx}
\usepackage{float}
\usepackage{setspace}
\numberwithin{equation}{section}
\theoremstyle{plain}

\newtheorem{theorem}{Theorem}

\newtheorem{lemma}[theorem]{Lemma}

\begin{document}

\title[Optimal Dividend in the Dual Risk Model]{\textbf{Optimal Dividends in the Dual Risk Model}\\ % Title
	Under a Stochastic Interest Rate}

\author{Zailei Cheng}
\address
{Department of Mathematics \newline
\indent Florida State University \newline
\indent 1017 Academic Way \newline
\indent Tallahassee, FL-32306 \newline
\indent United States of America}
\email{
zcheng@math.fsu.edu}

\date{\today 
%\textit{Revised:} 9 December 2015
}
%\subjclass[2010]{91B30; 91B70} %Risk theory, insurance; Stochastic models 
\keywords{dual risk model, optimal dividends, stochastic interest rate, Barrier strategy, Threshold strategy}

\begin{abstract}
Optimal dividend strategy in dual risk model is well studied in the literatures. But to the best of our knowledge, all the previous works assumes  deterministic interest rate. In this paper, we study the optimal dividends strategy in dual risk model, under a stochastic interest rate, assuming the discounting factor follows a geometric Brownian motion or exponential L\'evy process. We will show that closed form solutions can be obtained.
\end{abstract}

\maketitle

\section{Introduction}

In a classical risk model with dividend payment, the surplus of an insurance company can be written as:
\begin{equation*}
X_t^D=x+ct-S_t-D_t
\tag{1.1}\end{equation*}

Here $x$ is the initial surplus, $c>0$ is the premium rate, $S_t=\sum_{i=1}^{N_t}Y_i$ is a compound Poisson process, which can be interpreted as the sum of claims. $N_t$ is a Poisson process with rate $\lambda>0$, $Y_i$ are i.i.d. random variables with p.d.f. $p(x)$. $\{D_t\}$ is a dividend process. 

On the other hand, dual risk model \cite{Avanzi} is related to the wealth process of companies like petroleum companies and high tech companies. The surplus in this case can be written as:
\begin{equation*}
X_t^D=x-ct+S_t-D_t
\tag{1.2}\end{equation*}

Here $x$ is the initial surplus, $c>0$ is the rate of expenses, $S_t=\sum_{i=1}^{N_t}Y_i$ is a compound Poisson process, which can be interpreted as the value of future gains from an invention or discovery. $\{D_t\}$ is a dividend rate process.

The problem of optimal dividend was proposed by Bruno De Finetti in 1957. He suggested that a company would seek to find a strategy in order to maximize the accumulated value of expected discounted dividends up to the ruin time.

Many papers on optimal dividend problem in dual risk model have been published in recent years. Avanzi \cite{Avanzi} applied the barrier strategy to dual risk model and obtained the optimal dividend strategy. D.Yao \cite{Yao} worked on optimal dividend problem by constructing two categories of suboptimal models in dual risk model, one is the ordinary dual model without issuance of equity, the other one assumes that, by issuing new equity, the company never goes bankrupt. Cheung and Drekic \cite{Cheung} studied dividend moments in the dual risk model by  deriving the integro-differential equations for the moments of the total discounted dividends as well as the Laplace transform of the time of ruin. D.Peng \cite{Peng} considered the dual risk model with exponentially distributed observation time and constant dividend barrier strategy. A very recent work by Fahim and Zhu \cite{FZ} worked on asymptotic analysis for optimal dividends in dual risk model.

In this paper, we assume that in the compound Poisson process, the size of the jump $Y_i$ follows the exponential distribution, i.e. $p(x)=\beta e^{-\beta x},\beta>0$. 

The accumulated value of expected discounted dividends up to the ruin time becomes:
\begin{equation*}
V^D(x)=\mathbb{E}\left[\int_{0}^{\tau^D}e^{-\delta t}dD_t\right]
\tag{1.3}\end{equation*}

$\delta$ is a deterministic interest rate and $e^{-\delta t}$ is called the discounting factor, where $\tau^D=\inf\{t:X_t^D<0\}$ is the ruin time. People seek to find a dividend payment strategy $\{D_t^*\}$ so that 
\begin{equation*}
V(x)=\sup_{D\in\psi}\{V^D(x)\}=V^{D^*}(x)
\tag{1.4}\end{equation*}

The set of admissible strategies $\psi$ consists of non-negative, non-decreasing, adapted, c\`adl\`ag process. $V(x)$ is called the value function of optimal dividends problem \cite{Schmidli}.

To study these kinds of optimal control problem, The Hamilton-Jacobi-Bellman (HJB) equation is essential. The solution of the HJB equation is the value function with the optimal dividends. The HJB equation can be obtained by dynamical programming principle \cite{Albrecher,Schmidli}.

People have studied optimal dividends in classical risk model as well as in dual risk model under a deterministic interest rate \cite{Avanzi,Schmidli,Gerber,Ng,Albrecher}. Recently J.Eisenberg \cite{Eisenberg} published a paper on optimal dividends in the setting of a diffusion approximation of a classical risk model, i.e. Brownian motion with drift. The interest rate in this model also follows a Brownian motion with drift, which is stochastic. They found an explicit expression for the value function of the optimal strategy for both restricted dividends and unrestricted dividends. Also when I was preparing for this paper, I noticed that there is a very recent paper by J.Eisenberg and P.Kr\"uhner which uses the idea of optimal dividends in exponential L\'evy model \cite{EisenbergP}.

In the present paper, we will study the optimal dividends in dual risk model, under a stochastic interest rate.

%%%%%%%%%%%%%%%%%%%%%%%%%%%%%%%%%%%%%%%%%%%%%%%%%%%%

\section{Geometric Brownian Motion as a Discounting Factor}
We assume for the moment that the stochastic interest rate follows a Brownian motion with drift. So the discounting factor now becomes a geometric Brownian motion:
$$\exp\{-r-mt-\delta B_t\}$$
where $r>0$, $m>0$, $\delta\geq0$,
we assume 
$$m>\frac{\delta^2}{2}$$ throughout this section.\\ 
Given a strategy $D$, the return function, which is the accumulated value of expected discounted dividends up to the ruin time, is given by:
$$V^D(x,r)=\mathbb{E}\left[\int_{0}^{\tau^D}e^{-r-mt-\delta B_t}dD_t\right]$$

\subsection{Restricted Dividends}
In this case we only consider the dividend strategy $D_t$ that is bounded. i.e. $dD_t=U_tdt,U_t\in[0,\xi]$, where $\xi>0$ is a constant.
%\subsubsection{Proposition}

We abuse the notation by using $V^U$ to denote $V^D$.\\ 

\begin{lemma}
The return function $V^U(x,r)$ is bounded.
\end{lemma}

\begin{proof}
\begin{align*}
V^U(x,r)&=\mathbb{E}\left[\int_{0}^{\tau^U}e^{-r-mt-\delta B_t}U_t dt\right]\\
&\leq\mathbb{E}\left[\int_{0}^{\infty}e^{-r-mt-\delta B_t}\xi dt\right]\\
&=\xi e^{-r}\mathbb{E}\left[\int_{0}^{\infty}e^{-mt-\delta B_t}dt\right] \\
&=\xi e^{-r}\int_{0}^{\infty}e^{-mt}\mathbb{E}[e^{-\delta B_t}]dt\\
&=\xi e^{-r}\int_{0}^{\infty}e^{-(m-\delta^2/2)t}dt\\
&=\frac{\xi e^{-r}}{m-\delta^2/2}
\tag{2.1.1}\end{align*}
\end{proof}

The value function $V(x,r)=\sup_{U\in\psi}\{V^U(x,r)\}$.

The HJB equation corresponding to this problem is given by

\begin{equation*}
	-cV_x+mV_r+\frac{1}{2}\delta^2V_{rr}+\lambda \int_{0}^{\infty}[V(x+y,r)-V(x,r)]p(y)dy+\sup_{0\leq u\leq\xi}u\{e^{-r}-V_x\}=0
	\tag{2.1.2}\end{equation*}
	
\begin{lemma}
$V^\xi(x,r)$ obtained from $U_t\equiv\xi$ solves the HJB equation if $\frac{\xi\alpha}{m-\frac{\delta^2}{2}}\leq1$,\\
where 
$\alpha=\frac{-(\theta+\lambda-\beta c-\beta \xi)-\sqrt{(\theta+\lambda-\beta c-\beta \xi)^2+4\theta\beta(c+\xi)}}{2(c+\xi)}$, 
$\theta=m-\frac{\delta^2}{2}$
\end{lemma}

\begin{proof}
We set:

\begin{equation*}
\tau_x^{\xi,0}:=\inf\{t\geq0:x-(c+\xi)t+S_t=0\}
\tag{2.1.3}\end{equation*}

the corresponding return function is:
\begin{align*}
V^\xi (x,r) &=\xi \mathbb{E}\left[\int_{0}^{\tau_x^{\xi,0}}e^{-r-mt-\delta B_t}dt\right]\\
&=\xi \mathbb{E}\left[\int_{0}^{\tau_x^{\xi,0}}\mathbb{E}[e^{-r-mt-\delta B_t}]dt\right]\\
&=\frac{\xi e^{-r}}{m-\frac{\delta^2}{2}}\mathbb{E}[1-e^{-(m-\frac{\delta^2}{2})\tau_x^{\xi,0}}]
\tag{2.1.4}\end{align*}

Now suppose $m(x)=\mathbb{E}[e^{-\theta \tau_x^{\xi,0}}],\theta=m-\frac{\delta^2}{2}>0$,
$m(x)$ satisfies:
\begin{equation*}
-(c+\xi)m^\prime(x)+\lambda \int_{0}^{\infty}[m(x+y)-m(x)]p(y)dy-\theta m(x)=0
\tag{2.1.5}\end{equation*}

We conjecture that the solution is like $m(x)=Ae^{\alpha x}$, $\alpha<0$. Because ruin is immediate if $x=0$, so $m(0)=1$, i.e. $A=1$. Substitute into (2.1.5) we get: ($p(y)=\beta e^{-\beta y}$)

\begin{equation*}
(c+\xi)\alpha^2+(\theta+\lambda-\beta c-\beta \xi)\alpha-\theta \beta=0
\tag{2.1.6}\end{equation*}

Solve (2.1.6), $\alpha=\frac{-(\theta+\lambda-\beta c-\beta \xi)-\sqrt{(\theta+\lambda-\beta c-\beta \xi)^2+4\theta\beta(c+\xi)}}{2(c+\xi)}<0$

\begin{equation*}
V^\xi (x,r)=\frac{\xi e^{-r}}{m-\frac{\delta^2}{2}}(1-e^{\alpha x})
\tag{2.1.7}\end{equation*}

\begin{equation*}
V_x^\xi (x,r)=\frac{-\xi e^{-r}}{m-\frac{\delta^2}{2}}\alpha e^{\alpha x}>0
\tag{2.1.8}\end{equation*}

In particular, $V_x^\xi (r,x)\leq e^{-r}$ if $-\frac{\xi\alpha}{m-\frac{\delta^2}{2}}\leq1$. 

This means that $V^\xi (r,x)$ solves the HJB equation if $-\frac{\xi\alpha}{m-\frac{\delta^2}{2}}\leq1$.
\end{proof}

\begin{lemma}
\begin{equation*}
V(x,r)=e^{-r}F(x)=\left\{
\begin{array}{rcl}
e^{-r}F_1(x) &&  x>\hat{x}\\
e^{-r}F_2(x) &&  x\leq\hat{x}
\end{array}
\right.
\end{equation*}
solves the HJB equation if $-\frac{\xi\alpha}{m-\frac{\delta^2}{2}}>1$,\\
where $F_1(x)=Ae^{r_1x}+\frac{\xi}{m-\frac{\delta^2}{2}}$ and $F_2(x)=B(e^{s_1 x}-e^{s_2 x})$\\
$r_1$ is the negative solution of the equation:

\begin{equation*}
(c+\xi)x^2+\left[m+\lambda-\frac{\delta^2}{2}-\beta(c+\xi)\right]x-\beta(m-\frac{\delta^2}{2})=0
\end{equation*}\\
$s_1$ and $s_2$ are respectively the positive and negative solution of the function:

\begin{equation*}
cx^2+(m+\lambda-\frac{\delta^2}{2}-\beta c)x-\beta(m-\frac{\delta^2}{2})=0
\end{equation*}
\begin{equation*}
A=-\frac{\xi(\beta-r_1)}{\beta(m-\frac{\delta^2}{2})}\times \frac{s_1(\beta-s_2)e^{s_1\hat{x}}-s_2(\beta-s_1)e^{s_2\hat{x}}}{(s_1-r_1)(\beta-s_2)e^{s_1\hat{x}}-(s_2-r_1)(\beta-s_1)e^{s_2\hat{x}}}e^{-r_1\hat{x}}
\end{equation*}

\begin{equation*}
B=\frac{\xi(-r_1)}{\beta(m-\frac{\delta^2}{2})}\times \frac{(\beta-s_1)(\beta-s_2)}{(s_1-r_1)(\beta-s_2)e^{s_1\hat{x}}-(s_2-r_1)(\beta-s_1)e^{s_2\hat{x}}}
\end{equation*}
\begin{equation*}
\hat{x}=\frac{1}{s_1-s_2}ln\frac{s_2(s_2-r_1)(\beta-s_1)}{s_1(s_1-r_1)(\beta-s_2)}
\end{equation*}
\end{lemma}

\begin{proof}

If $-\frac{\xi\alpha}{m-\frac{\delta^2}{2}}>1$, according to (2.1.7), we conjecture that $V(x,r)=e^{-r}F(x)$. Substitute into (2.1.2),

\begin{equation*}
	-cF^\prime(x)-(m+\lambda-\frac{\delta^2}{2})F(x)+\lambda \int_{0}^{\infty}F(x+y)p(y)dy+\sup_{0\leq u\leq\xi}u(1-F^\prime(x))=0
	\tag{2.1.9}\end{equation*}

According to Page 99 of Schmidli \cite{Schmidli}, we need to solve the following two equations:

\begin{equation*}
	-cF_1^\prime(x)-(m+\lambda-\frac{\delta^2}{2})F_1(x)+\lambda \int_{0}^{\infty}F_1(x+y)p(y)dy+\xi(1-F_1^\prime(x))=0
	\tag{2.1.10}\end{equation*}

\begin{equation*}
	\begin{split}
		&-cF_2^\prime(x)-(m+\lambda-\frac{\delta^2}{2})F_2(x)+\lambda[\int_{0}^{\hat{x}-x}F_2(x+y)p(y)dy\\
		&+\int_{\hat{x}-x}^{\infty}F_1(x+y)p(y)dy]=0
	\end{split}
	\tag{2.1.11}\end{equation*}

$\hat{x}$ is a threshold point that

\begin{equation*}
	F(x)=\left\{
	\begin{array}{rcl}
		F_1(x) &&  x>\hat{x}\\
		F_2(x) &&  x\leq\hat{x}
	\end{array}
	\right.
	\tag{2.1.12}\end{equation*}

In (2.1.11) we write $\lambda \int_{0}^{\infty}F(x+y)p(y)dy=\lambda[\int_{0}^{\hat{x}-x}F_2(x+y)p(y)dy+\int_{\hat{x}-x}^{\infty}F_1(x+y)p(y)dy]$ because of the jump size $y$ in $F(x+y)$.

First we solve (2.1.10), proceeding like (2.1.5), we conjecture that:
\begin{equation*}
	F_1(x)=Ae^{r_1x}+\frac{\xi}{m-\frac{\delta^2}{2}}
	\tag{2.1.13}\end{equation*}

Note that $r_1<0$ because according to (2.1.1), $V^U(x,r)$ is bounded. $\frac{\xi}{m-\frac{\delta^2}{2}}$ is a particular solution of the inhomogeneous equation (2.1.10).

Substitute (2.1.13) into (2.1.10) we get $r_1$ to be the negative solution of the equation:

\begin{equation*}
	(c+\xi)x^2+[m+\lambda-\frac{\delta^2}{2}-\beta(c+\xi)]x-\beta(m-\frac{\delta^2}{2})=0
	\tag{2.1.14}\end{equation*}

Next we solve (2.1.11). Substitute (2.1.13) into (2.1.11), we get:

\begin{equation*}
	\begin{split}
		&-cF_2^\prime(x)-(m+\lambda-\frac{\delta^2}{2})F_2(x)+\lambda\int_{0}^{\hat{x}-x}F_2(x+y)p(y)dy\\
		&+\frac{\lambda A\beta}{\beta-r_1}e^{(r_1-\beta)\hat{x}+\beta x}+\frac{\lambda\xi}{m-\frac{\delta^2}{2}}e^{-\beta(\hat{x}-x)}=0
	\end{split}
	\tag{2.1.15}\end{equation*}

Note that by change of variable, $\int_{0}^{\hat{x}-x}F_2(x+y)p(y)dy$ can be written as $\int_{x}^{\hat{x}}F_2(u)p(u-x)du$. Also by applying the operator $(\frac{d}{dx}-\beta)$, we can get rid of the terms including $e^{\beta x}$. Then (2.1.15) becomes:

\begin{equation*}
	cF_2^{\prime\prime}(x)+(m+\lambda-\frac{\delta^2}{2}-\beta c)F_2^\prime(x)-\beta(m-\frac{\delta^2}{2})F_2(x)=0
	\tag{2.1.16}\end{equation*}

Noting that $F_2(0)=0$,

\begin{equation*}
	F_2(x)=B(e^{s_1 x}-e^{s_2 x})
	\tag{2.1.17}\end{equation*}

$s_1$ and $s_2$ are respectively the positive and negative solution of the function:

\begin{equation*}
	cx^2+(m+\lambda-\frac{\delta^2}{2}-\beta c)x-\beta(m-\frac{\delta^2}{2})=0
	\tag{2.1.18}\end{equation*}

Now we need to determine the constant $A,B,\hat{x}$.

Substitute (2.1.17) back into (2.1.15), we get:

\begin{equation*}
	\begin{split}
		&(\lambda+m-\frac{\delta^2}{2}+cs_1)Be^{s_1x}-(\lambda+m-\frac{\delta^2}{2}+cs_2)Be^{s_2x}\\
		&=\frac{\lambda\beta B}{\beta-s_1}e^{s_1x}-\frac{\lambda\beta B}{\beta-s_2}e^{s_2x}-\frac{\lambda\beta Be^{-(\beta-s_1)\hat{x}}}{\beta-s_1}e^{\beta x}\\
		&+\frac{\lambda\beta Be^{-(\beta-s_2)\hat{x}}}{\beta-s_2}e^{\beta x}+\frac{\lambda\beta Be^{-(\beta-r_1)\hat{x}}}{\beta-r_1}e^{\beta x}+\frac{\lambda\xi e^{-\beta\hat{x}}}{m-\frac{\delta^2}{2}}e^{\beta x}
	\end{split}
	\tag{2.1.19}\end{equation*}

Since the expression above holds for all $0\leq x\leq\hat{x}$, the sum of the coefficients of $e^{\beta x}$ must be zero.

\begin{equation*}
	B(\frac{e^{s_1\hat{x}}}{\beta-s_1}-\frac{e^{s_2\hat{x}}}{\beta-s_2})=\frac{Ae^{r_1\hat{x}}}{\beta-r_1}+\frac{\xi}{\beta(m-\frac{\delta^2}{2})}
	\tag{2.1.20}\end{equation*}

Also by the continuity condition, $F_1(\hat{x})=F_2(\hat{x})$

\begin{equation*}
	B(e^{s_1\hat{x}}-e^{s_2\hat{x}})=Ae^{r_1\hat{x}}+\frac{\xi}{m-\frac{\delta^2}{2}}
	\tag{2.1.21}\end{equation*}

By solving (2.1.20) and (2.1.21), we get:

\begin{equation*}
	A=-\frac{\xi(\beta-r_1)}{\beta(m-\frac{\delta^2}{2})}\times \frac{s_1(\beta-s_2)e^{s_1\hat{x}}-s_2(\beta-s_1)e^{s_2\hat{x}}}{(s_1-r_1)(\beta-s_2)e^{s_1\hat{x}}-(s_2-r_1)(\beta-s_1)e^{s_2\hat{x}}}e^{-r_1\hat{x}}
	\tag{2.1.22}\end{equation*}

\begin{equation*}
	B=\frac{\xi(-r_1)}{\beta(m-\frac{\delta^2}{2})}\times \frac{(\beta-s_1)(\beta-s_2)}{(s_1-r_1)(\beta-s_2)e^{s_1\hat{x}}-(s_2-r_1)(\beta-s_1)e^{s_2\hat{x}}}
	\tag{2.1.23}\end{equation*}

To calculate $\hat{x}$, since $V^\xi(x,r)$ achieves the maximum at $\hat{x}$, we need to maximize $A$ and $B$.

By looking at (2.1.23), we find that maximizing $B$ is equivalent to minimizing $(s_1-r_1)(\beta-s_2)e^{s_1\hat{x}}-(s_2-r_1)(\beta-s_1)e^{s_2\hat{x}}$,so

\begin{equation*}
	\hat{x}=\frac{1}{s_1-s_2}\ln\frac{s_2(s_2-r_1)(\beta-s_1)}{s_1(s_1-r_1)(\beta-s_2)}
	\tag{2.1.24}\end{equation*}

When maximizing $A$, we solve $\frac{\partial A}{\partial\hat{x}}=0$, after simple but tedious calculation we get  $s_1(s_1-r_1)(\beta-s_2)e^{s_1\hat{x}}-s_2(s_2-r_1)(\beta-s_1)e^{s_2\hat{x}}=0$. It also gives

\begin{equation*}
	\hat{x}=\frac{1}{s_1-s_2}\ln\frac{s_2(s_2-r_1)(\beta-s_1)}{s_1(s_1-r_1)(\beta-s_2)}
	\tag{2.1.25}\end{equation*}

The calculation of $\hat{x}$, on the contrary, proves that our calculation of $A$ and $B$ is correct.

Therefore, the value function:

\begin{equation*}
	V(x,r)=e^{-r}F(x)=\left\{
	\begin{array}{rcl}
		e^{-r}F_1(x) &&  x>\hat{x}\\
		e^{-r}F_2(x) &&  x\leq\hat{x}
	\end{array}
	\right.
	\tag{2.1.26}\end{equation*}
	
solves the HJB equation.
\end{proof}

Next let us provide a verification theorem to show that:

\begin{theorem}
The optimal strategy $U^*=\{U_t^*\}$ is

\begin{equation*}
U_t^* (x)=\xi\mathbb{I}_{\{X_t^{U^*}>\hat{x}\}}
\tag{2.1.27}\end{equation*}

\end{theorem}

Such a strategy is called the threshold strategy \cite{Schmidli,Gerber}. 

\begin{figure}[h]
\centering
\includegraphics[width=0.7\linewidth]{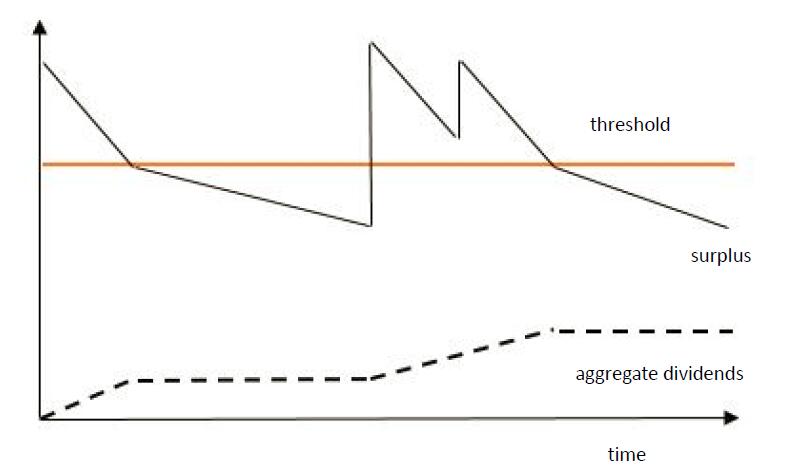}
\caption{Illustration of the threshold strategy in dual risk model, the $x$ axis denote the time evolution. The $y$ axis consists of three components, the red line is the threshold $\hat{x}$, the solid line is the surplus, the dashed line is the accumulated dividends. When the surplus is above the threshold, we pay the maximum dividend. When the surplus is below the threshold, we pay no dividend.}
\label{fig:threshold}
\end{figure}

\begin{proof}
	Suppose $U$ is an arbitrary admissible strategy and $\tau^U$ be the ruin time of $\{X_t^U\}$. Since $e^{-r}F(x)$ satisfies (2.1.2), by It\^o's formula:

	\begin{equation*}
		\begin{split}
			&e^{-r-m(t\wedge\tau^U)-\delta B_{t\wedge\tau^U}}F(X_{t\wedge\tau^U}^U)=e^{-r}F(x)\\&+
			\int_{0}^{t\wedge\tau^U}e^{-r-ms-\delta B_s}F^\prime(X_s^U)(-c-u)ds\\
			&+\lambda\int_{0}^{t\wedge\tau^U}e^{-r-ms-\delta B_s}\int_{0}^{\infty}[F(X_s^U+y)-F(X_s^U)]p(y)dyds\\
			&-\int_{0}^{t\wedge\tau^U}e^{-r-ms-\delta B_s}(m-\frac{\delta^2}{2})F(X_s^U)ds\\
			&+\delta\int_{0}^{t\wedge\tau^U}e^{-r-ms-\delta B_s}F(X_s^U)dB_s\\
			&\leq e^{-r}F(x)-\int_{0}^{t\wedge\tau^U}e^{-r-ms-\delta B_s}U_sds\\
			&+\delta\int_{0}^{t\wedge\tau^U}e^{-r-ms-\delta B_s}F(X_s^U)dB_s
		\end{split}
	\end{equation*}

	The equality holds when $U=U^*$
	
	Since $F$ is bounded and $\int_{0}^{t}[\mathbb{E}[e^{-r-ms-\delta B_s}]^2]ds<\infty$, so the last stochastic integral above is a martingale whose expectation euqals 0. Then we can get:
	
	\begin{equation*}
		\begin{split}
			\mathbb{E}[e^{-r-m(t\wedge\tau^U)-\delta B_{t\wedge\tau^U}}F(X_{t\wedge\tau^U}^U)]\\
			\leq e^{-r}F(x)-\mathbb{E}\left[\int_{0}^{t\wedge\tau^U}e^{-r-ms-\delta B_s}U_sds\right]
		\end{split}
		\tag{2.1.28}\end{equation*}
	
	If $\tau^U<t$, $F(X_{t\wedge\tau^U}^U)=F(0)=0$. Since $F(x)\leq\frac{\xi}{m-\frac{\delta^2}{2}}$, we have:
	
	\begin{equation*}
		\begin{split}
			&\mathbb{E}[e^{-r-m(t\wedge\tau^U)-\delta B_{t\wedge\tau^U}}F(X_{t\wedge\tau^U}^U)]\\
			&=\mathbb{E}[e^{-r-mt-\delta B_t}F(X_t^U)\mathbb{I}_{[\tau^U>t]}]\\
			&\leq\mathbb{E}[e^{-r-mt-\delta B_t}F(X_t^U)]\\
			&\leq\mathbb{E}[e^{-r-mt-\delta B_t}]\frac{\xi}{m-\frac{\delta^2}{2}}\\
			&=e^{-r-(m-\frac{\delta^2}{2})t}\frac{\xi}{m-\frac{\delta^2}{2}}
		\end{split}
	\end{equation*}
	
	Then we have:
	
	\begin{equation*}
		\lim_{t\to\infty}\mathbb{E}[e^{-r-m(t\wedge\tau^U)-\delta B_{t\wedge\tau^U}}F(X_{t\wedge\tau^U}^U)]=0
	\end{equation*}
	
	According to (2.1.28), 
	
	\begin{equation*}
		e^{-r}F(x)\left\{
		\begin{array}{rcl}
			\geq V^U(x,r) &&  U \ $arbitrary$\\
			=V^{U^*}(x,r) &&  U=U^*
		\end{array}
		\right.
	\end{equation*}
	
	So we conclude that :
	
	\begin{equation*}
		V(x,r)\geq V^{U^*}(x,r)=e^{-r}F(x)\geq\sup_{U\in\psi}V^U(x,r)=V(x,r)
	\end{equation*}

\end{proof}

\vspace{1em}
\vspace{1em}
\vspace{1em}
\subsection{Unrestricted Dividends}

Here we consider $D_t\in\psi$ without restriction.

The value function $V(x,r)=\sup_{D\in\psi}\{V^D(x,r)\}$.

This is a singular control problem, see e.g. chapter 8 in \cite{Fleming}  and the corresponding HJB equation is given by:

\begin{equation*}
	\begin{split}
		\max\{-cV_x+\lambda\int_{0}^{\infty}[V(x+y,r)-V(x,r)]p(y)dy\\
		+mV_r+\frac{1}{2}\delta^2 V_{rr} ; \ \ e^{-r}- V_x\}=0
	\end{split}
	\tag{2.2.1}\end{equation*}

Suppose the barrier strategy with parameter $b$ is applied \cite{Avanzi}, which means that no dividend is paid out if $X_t<b$ and the excess is paid out immediately as a dividend if $X_t>b$.

\begin{figure}[h]
\centering
\includegraphics[width=0.7\linewidth]{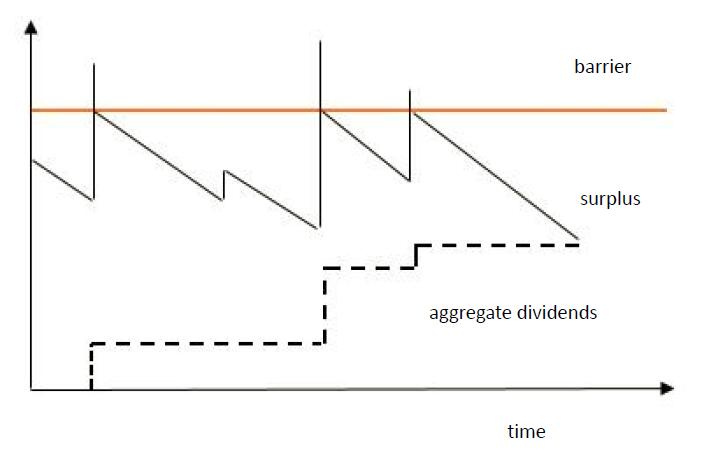}
\caption{Illustration of the barrier strategy in dual risk model, the $x$ axis denote the time evolution. The $y$ axis consists of three components, the red line is the barrier $b$, the solid line is the surplus, the dashed line is the accumulated dividends. When the surplus is above the threshold, we pay the excess of the surplus. When the surplus is below the threshold, we pay no dividend.}
\label{fig:barrier}
\end{figure}

\begin{lemma}
\begin{equation*}
V(x,r)=e^{-r}\left\{
\begin{array}{rcl}
K(e^{s_3 x}-e^{s_4 x}) &&  x\leq b\\
x-b+F(b,b) &&  x>b
\end{array}
\right.
\end{equation*}
solves the HJB equation,\\
where $s_3$, $s_4$ are respectively positive and negative solutions of the equation:

\begin{equation*}
cx^2+(m+\lambda-\frac{\delta^2}{2}-\beta c)x-\beta(m-\frac{\delta^2}{2})=0
\end{equation*}
\begin{equation*}
K=\frac{\lambda}{\beta}\times\frac{1}{(cs_3+m-\frac{\delta^2}{2})e^{s_3 b}-(cs_4+m-\frac{\delta^2}{2})e^{s_4 b}}
\end{equation*}
\begin{equation*}
F(x,b)=\frac{\lambda}{\beta}\times\frac{e^{s_3 x}-e^{s_4 x}}{(cs_3+m-\frac{\delta^2}{2})e^{s_3 b}-(cs_4+m-\frac{\delta^2}{2})e^{s_4 b}}, 0\leq x\leq b
\end{equation*}
\begin{equation*}
b=\frac{1}{s_3-s_4}\ln\frac{s_4(cs_4+m-\frac{\delta^2}{2})}{s_3(cs_3+m-\frac{\delta^2}{2})}
\end{equation*}
\end{lemma}

\begin{proof}

We try the Ansatz: 

\begin{equation*}
	V(x,r)=e^{-r}F(x,b)
	\tag{2.2.2}\end{equation*}

As ruin is immediate if $x=0$, so

\begin{equation*}
	F(0,b)=0
	\tag{2.2.3}\end{equation*}

First we consider the case when $x>b$, in this case \cite{Avanzi,Schmidli},

\begin{equation*}
	e^{-r}-V_x=0
	\tag{2.2.4}\end{equation*}

\begin{equation*}
	F(x,b)=x-b+F(b,b)
	\tag{2.2.5}\end{equation*}

Then we consider the case when $0<x\leq b$, in this case,

\begin{equation*}
	-cV_x+\lambda\int_{0}^{\infty}[v(x+y,r)-v(x,r)]p(y)dy
	+mV_r+\frac{1}{2}\delta^2 V_{rr}=0
	\tag{2.2.6}\end{equation*}

Substitute (2.2.2) into (2.2.6), we get:

\begin{equation*}
	\begin{split}
		&cF^\prime(x,b)+(m+\lambda-\frac{\delta^2}{2})F(x,b)-\lambda\int_{0}^{b-x}F(x+y,b)p(y)dy\\
		&-\lambda\int_{b-x}^{\infty}[x+y-b+F(b,b)]p(y)dy=0
	\end{split}
	\tag{2.2.7}\end{equation*}

Notice that:
\begin{equation*}
	\int_{a}^{\infty}p(y)(y-a)dy=\int_{a}^{\infty}(1-{P}(y))dy
	\tag{2.2.8}\end{equation*}

where ${P}(y)$ is the c.d.f of $p(y)$.

Then (2.2.7) can be rewritten as

\begin{equation*}
	\begin{split}
		&cF^\prime(x,b)+(m+\lambda-\frac{\delta^2}{2})F(x,b)-\lambda\int_{x}^{b}F(u,b)p(u-x)du\\
		&-\lambda\int_{b-x}^{\infty}[1-{P}(y)]dy-\lambda F(b,b)[1-{P}(b-x)]=0
	\end{split}
	\tag{2.2.9}\end{equation*}

Substitute $p(y)=\beta e^{-\beta y}$ and ${P}(y)=1-e^{-\beta y}$, apply the operator $(\frac{d}{dx}-\beta)$, (2.2.9) becomes:

\begin{equation*}
	cF^{\prime\prime}(x,b)+(m+\lambda-\frac{\delta^2}{2}-\beta c)F^\prime(x,b)-\beta(m-\frac{\delta^2}{2})F(x,b)=0
	\tag{2.2.10}\end{equation*}

\begin{equation*}
	F(x,b)=K(e^{s_3 x}-e^{s_4 x})
	\tag{2.2.11}\end{equation*}

$s_3$, $s_4$ are respectively positive and negative solutions of the equation:

\begin{equation*}
	cx^2+(m+\lambda-\frac{\delta^2}{2}-\beta c)x-\beta(m-\frac{\delta^2}{2})=0
	\tag{2.2.12}\end{equation*}

Substitute (2.2.11) back into (2.2.9) and set $x=b$, we get

\begin{equation*}
	K=\frac{\lambda}{\beta}\times\frac{1}{(cs_3+m-\frac{\delta^2}{2})e^{s_3 b}-(cs_4+m-\frac{\delta^2}{2})e^{s_4 b}}
	\tag{2.2.13}\end{equation*}

So 

\begin{equation*}
	F(x,b)=\frac{\lambda}{\beta}\times\frac{e^{s_3 x}-e^{s_4 x}}{(cs_3+m-\frac{\delta^2}{2})e^{s_3 b}-(cs_4+m-\frac{\delta^2}{2})e^{s_4 b}}, 0\leq x\leq b
	\tag{2.2.14}\end{equation*}

$F(x,b)$ is maximized at $b$, so $\frac{\partial F}{\partial b}=0$

\begin{equation*}
	b=\frac{1}{s_3-s_4}\ln\frac{s_4(cs_4+m-\frac{\delta^2}{2})}{s_3(cs_3+m-\frac{\delta^2}{2})}
	\tag{2.2.15}\end{equation*}

In conclusion, the value function:

\begin{equation*}
	V(x,r)=e^{-r}\left\{
	\begin{array}{rcl}
		K(e^{s_3 x}-e^{s_4 x}) &&  x\leq b\\
		x-b+F(b,b) &&  x>b
	\end{array}
	\right.
	\tag{2.2.16}\end{equation*}

solves the HJB equation. 
\end{proof}

Next we need to prove a verification theorem:

\begin{theorem}
The optimal strategy $D^*$ is to pay out any capital greater than $b$, i.e. $D_t^*=\max\{\sup_{0\leq s\leq\tau\wedge t}X_s-b,0\}$, where $X_s=x-cs+S_t$ and $\tau$ denotes the ruin time under strategy $D^*$.
\end{theorem}

\begin{proof}
	If $x>b$, the proof is obvious because $V(x,r)=e^{-r}[x-b+F(b,b)]$. If $x\leq b$, the process $D^*$ is continuous and increasing and therefore of bounded variation. According to It\^o's formula, under the strategy $D^*$:
	
	\begin{equation*}
		\begin{split}
			&e^{-r-m(t\wedge\tau^{D^*})-\delta B_{t\wedge\tau^{D^*}}} F(X_{t\wedge\tau^{D^*}}^{D^*})=e^{-r}F(x)\\+
			&\int_{0}^{t\wedge\tau^{D^*}}e^{-r-ms-\delta B_s}F^\prime(X_s^{D^*})(-c)ds\\
			&+\lambda\int_{0}^{t\wedge\tau^{D^*}}e^{-r-ms-\delta B_s}\int_{0}^{\infty}[F(X_s^{D^*}+y)-F(X_s^{D^*})]p(y)dyds\\
			&-\int_{0}^{t\wedge\tau^{D^*}}e^{-r-ms-\delta B_s}(m-\frac{\delta^2}{2})F(X_s^{D^*})ds\\
			&-\int_{0}^{t\wedge\tau^{D^*}}e^{-r-ms-\delta B_s}F^\prime(X_s^{D^*})D_s^*ds\\
			&+\delta\int_{0}^{t\wedge\tau^{D^*}}e^{-r-ms-\delta B_s}F(X_s^{D^*})dB_s\\
			&= e^{-r}F(x)-\int_{0}^{t\wedge\tau^{D^*}}e^{-r-ms-\delta B_s}D_s^*ds\\
			&+\delta\int_{0}^{t\wedge\tau^{D^*}}e^{-r-ms-\delta B_s}F(X_s^{D^*})dB_s
		\end{split}
	\end{equation*}
	
	The last step is because $V$ satisfies (2.2.1) and $D^*$ only increases at points where $x_s^{D^*}=b$, i.e., $F^\prime(X_s^{D^*})=1$. Also because $F$ is bounded when $x<b$ and $\int_{0}^{t}[\mathbb{E}[e^{-r-ms-\delta B_s}]]^2ds<\infty$, so the last stochastic integral above is a martingale with expectation 0. Taking expectation on both sides of the above equality yields:
	
	\begin{equation*}
		e^{-r}F(x)=\mathbb{E}\left[e^{-r-m(t\wedge\tau^{D^*})-\delta B_{t\wedge\tau^{D^*}}} F(X_{t\wedge\tau^{D^*}}^{D^*})+\int_{0}^{t\wedge\tau^{D^*}}e^{-r-ms-\delta B_s}D_s^*ds\right]
	\end{equation*}
	
	Then we let $t\to\infty$, $e^{-r}F(x)=V(x,r)=V^{D^*}(x,r)$ because $F(X_\tau)=0$.

\end{proof}

\section{Exponential L\'evy Process as a Discounting Factor}

In this section the discounting rate is assumed to follow an exponential L\'evy process $\exp(-r-mt-X_t)$, where $r>0$, $m>0$ $X_t$ is a L\'evy process with characteristic triplet $(\delta, \gamma, \nu)$, i.e., the characteristic function of $X_t$ has the following L\'evy-Khinchin representation: \cite{Cont}.

\begin{equation*}
	\mathbb{E}[e^{izX_t}]=\exp t\phi(z)
\end{equation*}

\begin{equation*}
	\phi(z)=-\frac{\delta^2z^2}{2}+i\gamma z+\int_{-\infty}^{+\infty}(e^{izx}-1-izx\mathbb{I}_{|x|\leq1}\nu(dx))
\end{equation*}

Its infinitesimal generator is given by:

\begin{equation*}
	Lf(x)=\frac{\delta^2}{2}\frac{\partial^2f}{\partial x^2}+\gamma\frac{\partial f}{\partial x} +\int\nu(dy)[f(x+y)-f(x)-y\mathbb{I}_{|y|\leq1}\frac{\partial f}{\partial x}(x)]
\end{equation*}

Here $\delta>0$ and $\gamma$ are real constants and $\nu$ are p.d.f of jump size $y$.

\subsection{Restricted Case}

The HJB equation corresponding to this case is:

\begin{equation*}
	\begin{split}
		-cV_x+\lambda\int_{0}^{\infty}[V(x+y,r)-V(x,r)]p(y)dy+(m+\gamma)V_r+\frac{\delta^2}{2}V_{rr}\\
		+\int_{0}^{\infty}[V(x,r+z)-V(x,r)-z\mathbb{I}_{\{|z|\leq 1\}}V_r]\nu(z)dz+\sup_{0\leq u\leq\xi}u(e^{-r}-V_x)
	\end{split}
	\tag{3.1.1}\end{equation*}

Proceeding like Sec. 2.1, substitute $V(x,r)=e^{-r}F(x)$ into (3.1.1), we have to solve the following two equations:

\begin{equation*}
	\begin{split}
		-cF_1^\prime(x)-(m+\lambda+\nu+1-\frac{\delta^2}{2}-k-l)F_1(x)\\
		+\lambda \int_{0}^{\infty}F_1(x+y)p(y)dy+\xi(1-F_1^\prime(x))=0
	\end{split}
	\tag{3.1.2}\end{equation*}

\begin{equation*}
	\begin{split}
		&-cF_2^\prime(x)-(m+\lambda+\nu+1-\frac{\delta^2}{2}-k-l)F_2(x)\\
		&+\lambda[\int_{0}^{\hat{x}-x}F_2(x+y)p(y)dy+\int_{\hat{x}-x}^{\infty}F_1(x+y)p(y)dy]=0
	\end{split}
	\tag{3.1.3}\end{equation*}

where $k=\int_{0}^{\infty}e^z\nu(z)dz$ and $l=\int_{0}^{\infty}z\mathbb{I}_{\{|z|\leq 1\}}\nu(z)dz$ are two constants.

$\hat{x}$ is a boundary point that

\begin{equation*}
	F(x)=\left\{
	\begin{array}{rcl}
		F_1(x) &&  x>\hat{x}\\
		F_2(x) &&  x\leq\hat{x}
	\end{array}
	\right.
	\tag{3.1.4}\end{equation*}

Now if we set $\bar{m}=m+\nu+1-k-l$, our results would be the same as Sec. 2.1, just replacing the $m$ in Sec. 2.1 by $\bar{m}$. The precondition is that $\bar{m}-\frac{\delta^2}{2}>0$. 

\subsection{Unrestricted Case}

The HJB equation corresponding to this case is:

\begin{equation*}
	\begin{split}
		&\max\{-cV_x+\lambda\int_{0}^{\infty}[V(x+y,r)-V(x,r)]p(y)dy+(m+\nu)V_r+\frac{1}{2}\delta^2V_{rr}\\
		&+\int_{0}^{\infty}[V(x,r+z)-V(x,r)-z\mathbb{I}_{\{|z|\leq 1\}}V_r]\nu(z)dz; \ \ e^{-r}- V_x\}=0
	\end{split}
	\tag{3.2.1}\end{equation*}

Similar to Sec. 3.1, if we set $\bar{m}=m+\nu+1-k-l$, our results would be the same as Sec. 2.2, just replacing the $m$ in Sec. 2.2 by $\bar{m}$. The precondition is that $\bar{m}-\frac{\delta^2}{2}>0$.

%------------------------------------------------

\section{Conclusion}

From the discussions above, we find the optimal dividends strategies in dual risk models under stochastic interest rates, assuming the discounting factor is a geometric Brownian motion or exponential L\'evy process.  The strategies are analogous to the ones discussed by Eisenberg \cite{Eisenberg}. In restricted case, where the dividends payment is bounded, the optimal strategy is a threshold strategy. In unrestricted case, where the dividends payment is unbounded, the optimal strategy is a barrier strategy. In both cases, closed form solutions have been obtained.
\\
\\
{\bf Acknowledgments}\\

The author would like to express his special thanks to his advisor, Dr. Lingjiong Zhu, for his many useful comments and suggestions on this research.\\

%%%%%%%%%%%%%%%%%%%%%%%%%%%%%%%%%%%%%%%%%%%%%%%%%%%%%%

\end{document}